\documentclass[a4paper,UKenglish,cleveref, autoref, thm-restate]{lipics-v2021}

\usepackage{amsmath}
\usepackage{amssymb}
\usepackage{stmaryrd}
\usepackage{xspace}

\title{Scenes: A Meta-Logical Algebra for Mutable State} 

\author{Simon Foster}{University of York}{simon.foster@york.ac.uk}{https://orcid.org/0000-0002-9889-9514}{}

\author{Carlos Isasa}{University of Aarhus}{}{https://orcid.org/0000-0002-5035-8559}{}

\author{Christian Pardillo Laursen}{University of York}{}{https://orcid.org/0000-0001-7838-2764}{}

\authorrunning{Simon Foster, Carlos Isasa, Christian Pardillo Laursen}

\Copyright{Simon Foster, Carlos Isasa, Christian Pardillo Laursen}

\ccsdesc[500]{Theory of computation~Semantics and reasoning}

\keywords{Program verification, State modelling, Isabelle, Formal verification} 

\category{} 

\relatedversion{} 

\acknowledgements{We thank Brijesh Dongol and Georg Struth for their helpful feedback on an early version of this work.}

\nolinenumbers 

\hideLIPIcs

\newcommand{\lget}{\textit{\textsf{get}}\xspace}
\newcommand{\lput}{\textit{\textsf{put}}\xspace}

\newcommand{\lto}{\Longrightarrow}
\newcommand{\lone}{\mathbf{1}}
\newcommand{\lzero}{\mathbf{0}}
\newcommand{\lcomp}{\mathop{\fatsemi}}
\newcommand{\lindep}{\mathop{\bowtie}}
\newcommand{\lequiv}{\mathop{\approx}}
\newcommand{\lsublens}{\mathop{\preceq}}
\newcommand{\lplus}{\mathop{\oplus}}
\newcommand{\scompat}{\mathop{\#}}
\newcommand{\sceneq}[1]{\simeq_{#1}}

\newcommand{\unrest}{\mathop{\sharp}}

\newcommand{\view}{\mathcal{V}}
\newcommand{\src}{\mathcal{S}}

\newcommand{\mrg}[1][]{\textit{\textsf{mrg}}_{\scriptscriptstyle #1}}
\newcommand{\mrgon}[3]{#1 \lhd #2 \mathop{on} #3}
\newcommand{\defs}{\triangleq}

\newcommand{\num}{\mathbb{Z}}
\newcommand{\basis}{\mathcal{B}}
\newcommand{\real}{\mathbb{R}}

\newcommand{\isalogo}{\includegraphics[width=9pt]{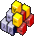}}
\newcommand{\isalink}[1]{\hfill \href{#1}{\isalogo}}

\begin{document}

\maketitle

\begin{abstract}
Modelling of mutable state spaces and precisely describing how variables are manipulated in a program is a fundamental problem in compositional verification. Though we can make use of the embedded abstract syntax of a program for such analysis, this runs contrary to the shallow-embedding approach, and hampers efficient proof automation. On the other hand, lenses and prisms provide an elegant algebraic foundation for modelling state, which provide sufficient structure to provide meta-logical program analysis, but without requiring a deep embedding. Nevertheless lenses, as complex algebraic objects, cannot easily be combined, complemented, or collected in sets. In this paper we contribute an accompanying algebraic structure called the \textit{scene}, which allows us to characterise the set of variables, or coordinates, in a state space. Scenes intuitively correspond to sets of lenses, but like lenses they are purely semantic algebraic objects. We demonstrate that scenes provide us with sufficient structure to characterise the lens-based meta-logical properties, like independence and equivalence. Moreover, we introduce the notion of a scene space, analogous to a vector space, which allows us to recover a set-like algebraic structure. Finally, we show how scenes allow us to characterise the free and bound variables of expressions and programs, without any need for syntax, and demonstrate their use for reasoning about programs by deriving reasoning principles for the parallel composition operator.
\end{abstract}

\newpage

\section{Introduction}
\label{sec:intro}
Modelling the mutable state of a program is a fundamental problem in verification~\cite{Back1998,Schirmer2009,Dongol19}. Though a naive model for state can be given through functions from variables to values (e.g. $N \rightarrow V$), in reality state spaces and stores are more complex and multi-faceted than this. For example, they can consist of concurrency, hierarchy, and addresses. In previous work~\cite{Foster2020-IsabelleUTP,HuertaFGSLH2024}, we have used lenses --- an algebraic structure that pairs a \lget and \lput function~\cite{Foster07} --- to provide a flexible and powerful foundation for modelling state. 

Lenses significantly narrow the gap between deep and shallow semantic embeddings~\cite{Foster2020-IsabelleUTP}. On the one hand, they do not require the formalisation of names and abstract syntax, which ensures that lens-based verification tools can maximise proof automation, through reuse of host logic facilities, as in a shallow embedding. On the other hand, lenses also provide sufficient structure on the store so that the semantics, rather than the abstract syntax, provides sufficient information for laws of programming, akin to a deep embedding. With them we can encode \textit{meta-logical operators}, that is those operators that characterise program properties, such as substitutions for variables, whether an expression depends on a variable or not, and whether a program modifies a variable. This has allowed us to develop several verification tools with lenses as the foundation, including concurrent systems~\cite{Foster17-ReactiveContracts,Foster2024ITrees}, hybrid systems~\cite{HuertaFGSLH2024}, probabilistic systems~\cite{Ye2024Probabilistic}, and robotic state machines~\cite{Foster2020-AUV,Yan2024ZMachines}.

However, whilst lenses have demonstrable utility, they lack the ability to describe collections of variables, and the typical set operators (membership, union, etc.). This is required, for example, so that we can compute the set of free variables in an expression, and so that we can manipulate a program's frame -- the set of variables that a program modifies~\cite{platzer_complete_2017}. Almost every formal verification approach requires that we can model sets of variables, in tools as diverse as Event-B~\cite{Abrial2010EventB} and KeYmaera X~\cite{KeYmaera}. The research question that this paper tackles is: how can we effectively model sets of program variables in the context of a shallow embedding?

Our answer to this question is \emph{scenes}, a simple yet powerful algebraic structure that allows us to semantically encode a set of variables. Scenes are essentially merge operators, which allow us to combine two program states together, according to a pre-defined division of the state variables. As with lenses, scenes do not depend on a particular state space model, but can be applied to any type that can be subdivided into variables. The intuition of scenes as sets requires that we can also formalise \emph{scene spaces}, which characterises foundational consistency properties for state spaces. Thus, scenes allow us to characterise a region of the store semantically, which in turn provides the basis for compositional reasoning by allowing us to precisely capture how a program modifies the state. We use scenes to derive results from Platzer's differential dynamic logic~\cite{Platzer18,platzer_complete_2017}, which are foundational to reasoning about concurrent hybrid programs.

Our theory does not require dependent types, and has been mechanised in the Isabelle proof assistant~\cite{Isabelle}. Every significant definition and theorem is accompanied by a link to our Git repository with the mechanised artefact \isalink. However, our results are also applicable to other proof assistants, such as Lean~\cite{moura2015lean} and Rocq~\cite{coq2024rocq}.

In Section~\ref{sec:lenses} we review preliminaries on lenses and state spaces. In Section~\ref{sec:scenes} we introduce the scene as an algebraic structure, and describe the fundamental operators. In Section~\ref{sec:scene_spaces} we introduce scene spaces, which characterise consistent state spaces, and give rise to a complete Boolean algebra. In Section~\ref{sec:programs} we apply all of the aforementioned theory in mechanising free variables and bound variables for expressions and programs, and derive key theorems from Platzer's work. In Section~\ref{sec:related} we consider related work and in Section~\ref{sec:concl} we conclude.

\section{Preliminaries: Lenses and State Spaces}
\label{sec:lenses}
Lenses are used to model the mutation of variables in the program state. A lens $X : \view \lto \src$ is a structure $(\view, \src, \lget, \lput)$, where $\lget_X : \src \to \view$ and $\lput_X : \view \to \src \to \src$, $\view$ is the view, and $\src$ is the source. Intuitively, a lens describes a $\view$-shaped region of a source type, $\src$. They can capture either the individual variables of type $\view$, or several variables. We constrain lenses to obey a number of algebraic laws, which describe how a single variable may be updated:
$$\lget~(\lput~v~s) ~=~ v \qquad \lput~v~(\lput~v'~s) ~=~ \lput~v~s \qquad \lput~(\lget~s)~s ~=~ s$$

\noindent For example, if $\src \defs \view_1 \times \view_2 \times \cdots \times \view_n$ then we can form a lens for each projection $\pi_i : \view_i \lto \src$ for $i \in \{1\cdots n\}$. In general, state spaces can be arbitrarily complex with hierarchy and nested data structures. Thus, lenses may be independent from each other, meaning they do not interfere, or one lens may be part of another. We characterise independence below:

\begin{definition} Lenses $X : V_1 \lto \src$ and $Y : V_2 \lto \src$ are independent, written $X \lindep Y$, provided that $\lput_X~v_1 \circ \lput_Y~v_2 = \lput_Y~v_2 \circ \lput_X~v_1$, for any $v_1 \in V_1$ and $v_2 \in V_2$.
\end{definition}

\noindent 
In the $n$-ary product space, $\pi_i$ and $\pi_j$ are independent $\pi_i \lindep \pi_j$ precisely when $i \neq j$. We can compose lenses $X \lcomp Y$ when the source of $X$ matches the view of $Y$~\cite{Foster2020-IsabelleUTP}.

We also define $\lone \defs (A, A, \lambda s.\,s, \lambda~v~s.\, s)$, which provides a left and right unit for $\lcomp$, and covers the whole state space. It is an example of a bijective lens, that is a lens that satisfies $\lput~s~(\lget~s') = s'$. Bijective lenses relate isomorphic representations of the same state space. 

We can now characterise the part-of relationship using a heterogeneous preorder on lenses:

\begin{definition} 
    Lens $X : V_1 \lto \src$ is a sublens $Y : V_2 \lto \src$, written $X \lsublens Y$, if there exists $Z : V_1 \lto V_2$, such that $X = Z \lcomp Y$.
\end{definition}

\noindent Intuitively, $X \lsublens Y$ when $Y$ encompasses all the variables in $X$. For example, $\lone$ encompasses all of the state, and so $X \lsublens \lone$. Moreover, we can define a heterogeneous equivalence: $X \lequiv Y \leftrightarrow (X \lsublens Y \land Y \lsublens X)$, for when two lenses encompass the same region. For example, if $X \approx \lone$, it means that $X$ covers all the state, and so is bijective.

In addition, we can compose two lenses in parallel using the lens sum operator:

\begin{definition}
    Given independent lenses $X : V_1 \lto \src$ and $Y : V_2 \lto \src$, the summation of two lenses is defined $X \lplus Y = (V_1 \times V_2, \src, \lambda s.\, (\lget_X~s, \lget_Y~s), \lambda (v_1, v_2).\, \lput_X~v_1 \circ \lput_Y~v_2)$.
\end{definition}

\noindent Operationally, $X \lplus Y$ corresponds to a simultaneous update to the coordinates characterised by $X$ and $Y$. However, it also allows us to characterise the region of the source covered by both $X$ and $Y$. Using this, we can now formally characterise a state space:

\begin{definition} \label{def:state-space}
    A state space is a family of lenses $\mathcal{A} = \{x_1, x_2, \cdots, x_n\}$, such that $x_i \lindep x_j$ for $i \neq j$, which form a span, that is $\left(\sum_{i \in \{1\cdots n\}} x_i\right) \lequiv \lone$.
\end{definition}
Each $x_i \in \mathcal{A}$ is called a \emph{basis lens}. Analogous to a vector space, combinations of basis lenses allow us to address any location. A state is then a set of maplets $[x_1 \leadsto v_1, x_2 \leadsto v_2, \cdots, x_n \leadsto v_n]$. However, we do yet not have the mathematical machinery necessary to formally characterise or construct such bases, which is the objective of sections~\ref{sec:scenes} and \ref{sec:scene_spaces}.

Each basis lens can also be considered as a \emph{coordinate} (or set thereof) that forms part of a coordinate system. The summation of the lenses in $\mathcal{A}$ yields a bijective lens that is analogous to a \emph{coordinate chart}~\cite{lee2012introduction}. Specifically, $\mathcal{X} = (\sum_{i \in \{1\cdots n\}} x_i)$ is a bijective lens of type $(\prod_{i \in \{1\cdots n\}} V_i) \lto \src$. Since $\mathcal{A}$ forms a span,  $\lget_\mathcal{X}$ transforms $\src$ into a flat product space, similar to how a coordinate chart $\varphi : U \rightarrow \mathbb{R}^n$ transforms a manifold into a Euclidean space. 

As an example, consider a type for Cartesian coordinates in a 2-dimensional plane \textit{coord} (i.e. $\real \times \real$). We define $xc, yc : \real \lto \textit{coord}$, which are independent ($xc \lindep yc$), and form a span ($xc \lplus yc \lequiv \lone$). They form a state space and a consistent coordinate system. We can also consider an alternative state space on the \textit{coord} type, as shown below. 
\begin{example}[Polar Coordinates Lenses] \label{ex:polarlens} \isalink{https://github.com/isabelle-utp/Optics/blob/f48265c6c05a4c7dc6951ecc29815369b8ae9377/Coordinate_Lenses_Example.thy}

$$\begin{array}{rcllrcl}
\textit{radius-of}(x, y)  &\defs& sqrt (x^2 + y^2)  &~~& 
\textit{angle-of}(x, y)   &\defs& \textit{arctan}(y / x)
 \\
\textit{xc-of}(r, \theta) &\defs& r \cdot cos(\theta) &~~&
\textit{yc-of}(r, \theta) &\defs& r \cdot sin(\theta)
\end{array}$$

\vspace{-4ex}

\begin{align*}
\textit{radius} &\defs (\textit{radius-of}, (\lambda (x, y)~r.\, (\textit{xc-of} (r, \textit{angle-of} (x, y)), \textit{xc-of} (r, \textit{angle-of} (x, y))))) \\
\textit{angle} &\defs (\textit{angle-of}, (\lambda (x, y)~\theta.\, (\textit{xc-of} (\textit{radius-of} (x, y), \theta), \textit{yc-of} (\textit{radius-of} (x, y), \theta))))
\end{align*}
\end{example}
We consider only strictly positive coordinates in the first quadrant 
by defining $pcoord \defs \real_{>0} \times \real_{>0}$, and define lenses $\textit{radius} : \real_{>0} \lto \textit{pcoord}$ and $\textit{angle}: \mathbb{A} \lto \textit{pcoord}$. These convert Cartesian coordinates into polar coordinates using trigonemtric identities. These lenses obey the three lens laws\footnote{This primarily follows because we have removed the origin.}. Moreover, $\textit{radius} \lindep \textit{angle}$, since we can modify either coordinate without affecting the other, and also $\textit{radius} \lplus \textit{angle} \lequiv \lone$, confirming that every Cartesian coordinate has a polar coordinate. Of course, these two coordinate systems are not compatible: we cannot switch arbitrarily between them without corrupting the state space. 

We have seen that lenses allow us to characterise a state space. However, they are complex algebraic objects to deal with, due to the presence of two sorts: the view $\view$ and the source $\src$. The preorder ($\lsublens$) and equality relation ($\lequiv$) are both heterogeneous, and so equational rewriting cannot be automated in proofs assistants like Isabelle. Moreover, though we can describe the family $\mathcal{A}$, in reality we cannot collect a set of lenses in a set within a typed logic, since every element of a set must have the same type. Therefore, in the next section we introduce a related algebraic object that will allow us to overcome these issues.

\section{Scenes}
\label{sec:scenes}
In this section we introduce a novel algebraic structure called a \emph{scene}. In contrast to lenses, scenes allow us to formally model sets of variables or coordinates, in a given state space.


\begin{definition} \label{def:scene} A scene is a structure $A \defs (\src, \mrg)$ where $\src \neq \emptyset$ and $\mrg : \src \to \src \to \src$ is a binary operator
  called the merge function. For any $x, y, z : \src$, it obeys the following laws\footnote{In Isabelle/HOL, these are called ``idempotent scenes'', but for this paper this distinction is irrelevant.}:
  \begin{align*}
    \mrg~(\mrg~x~y)~z ~=~ \mrg~x~z \qquad \mrg~x~(\mrg~y~z) ~=~ \mrg~x~z \qquad \mrg~x~x ~=~ x
  \end{align*}
  We write $Scene(\src)$ for the set of scenes over $\src$, that is, a family of merge functions over $\src$. We write $\mrg[A]$ and $\mrg[B]$ to distinguish the merge functions from several scenes. \isalink{https://github.com/isabelle-utp/Optics/blob/f48265c6c05a4c7dc6951ecc29815369b8ae9377/Scenes.thy\#L19}
\end{definition}


\noindent Intuitively, $\mrg$ merges two states, analogous to merge operators in programming languages. Given a state $[x_1 \leadsto v_1, x_2 \leadsto v_2, \cdots ]$ then a merge function $\mrg~s_1~s_2$ replaces the value of a subset of the variables in $s_1$ with those in $s_2$. 
Scenes therefore semantically describe a subset of the state variables without capturing the view, which yields a simpler algebraic structure.

We give an example family of scenes based on total functions.

\begin{definition}[Function Domain Scene] We assume non-empty sets $A$ and $B$, which provide the domain and range types for a
  function, and a set $D \subseteq A$, and then define \isalink{https://github.com/isabelle-utp/Optics/blob/f48265c6c05a4c7dc6951ecc29815369b8ae9377/Scenes.thy\#L638}
  $$\textsf{fds}^A_B(D) \defs \left(A \to B, \lambda f~g.\, \lambda x.\, 
    \begin{cases}
        f(x) & \text{if}~ x \in D \\
        g(x) & \text{otherwise}
    \end{cases}\right)$$
  This scene overrides $f$ with $g$ when an input value $x$ is a member of $D$. It is easy to show that the
  overriding function obeys the three laws in Definition~\ref{def:scene}. \qed
\end{definition}
Consider an example instantiation, where $A = \{a,b,c\}$, $B = \num$, and $D = \{a, b\}$. Then, we have $\mrg~\{a \mapsto 1, b \mapsto 2, c \mapsto 3\}~\{a \mapsto 7, b \mapsto 8, c \mapsto 9\} = \{a \mapsto 7, b \mapsto 8, c \mapsto 3\}$, which takes the values of the keys in $D$ from the first function, and the remaining values from the second.

Every scene induces an equivalence relation on states:

\begin{definition}[State Equivalence] $s_1 \sceneq{F} s_2 \leftrightarrow (\mrg[F]~s_1~s_2) = s_1$ \isalink{https://github.com/isabelle-utp/Optics/blob/f48265c6c05a4c7dc6951ecc29815369b8ae9377/Scenes.thy\#L97}
\end{definition}
States $s_1$ and $s_2$ are equivalent modulo $F$ when merging the two states yields $s_1$. This means that the region of the state space described by $F$ has the same valuation in both $s_1$ and $s_2$. For any given scene $F$, we show that $(\sceneq{F})$ is an equivalence relation on $\src$: it is reflexive, symmetric, and transitive. We sometimes also use the notation $s_1 \sceneq{} s_2 \mathop{\textit{on}} F$ for $s_1 \sceneq{F} s_2$. This relation has a central role in program verification, since it allows us to reason about the \textit{frame} of an expression or program: how much of these state is changed.

Every lens $X : V \lto \src$ has an underlying scene on $\src$ that captures the region of the state space viewed by $X$. The scene of a (total) lens $X$ is defined as follows: 
\begin{definition}[Lens Scene]
$X^\sim \defs (\lambda s_1~s_2.\, \lput_X~s_1~(\lget_X~s_2))$ \isalink{https://github.com/isabelle-utp/Optics/blob/f48265c6c05a4c7dc6951ecc29815369b8ae9377/Scenes.thy\#L452}
\end{definition}
This updates $s_1$ by copying the data located at $X$ from the corresponding data in $s_2$. It corresponds to the region of $\src$ characterised by $X$. Any lens satisfying the three lens laws yields a scene satisfying the three scene laws. We can therefore overload the lens operators (e.g. $\lindep$) to apply to both lenses and scenes. As for lenses, we can also define independence:

\begin{definition}[Independence] Scenes $F$ and $G$ over $\src$ are independent, written $F \lindep G$, provided that, for any $s_1, s_2, s_3 \in \src$, they satisfy $G~(F~s_1~s_2)~s_3 = F~(G~s_1~s_3)~s_2$. \isalink{https://github.com/isabelle-utp/Optics/blob/f48265c6c05a4c7dc6951ecc29815369b8ae9377/Scenes.thy\#L104}
\end{definition}
Here, two sets of updates are applied to $s_1$: $F$ is used to merge with updates from $s_2$ and $G$ is used to merge the updates from $s_3$. If these updates commute, then it means that there is no interference, and so they are independent. This definition also subsumes lens independence:

\begin{theorem} $X^\sim \lindep Y^\sim$ if and only if $\lput_X~v_1 \circ \lput_Y~v_2 = \lput_Y~v_2 \circ \lput_X~v_1$. \isalink{https://github.com/isabelle-utp/Optics/blob/f48265c6c05a4c7dc6951ecc29815369b8ae9377/Scenes.thy\#L488}
\end{theorem}
We are therefore justified in substituting lens independence for scene independence. We can also prove an important theorem linking lens equivalence and scene equality:

\begin{theorem}[Lens equivalence is Scene equality] $X \lequiv Y \longleftrightarrow X^\sim = Y^\sim$ \isalink{https://github.com/isabelle-utp/Optics/blob/f48265c6c05a4c7dc6951ecc29815369b8ae9377/Scenes.thy\#L580} \end{theorem}

\noindent Lens equivalence can by characterised by showing that the two underlying scenes are equal, using the standard homogeneous equality. This means that we can substitute the heterogeneous lens equivalence with a homogeneous scene equivalence, meaning the simplifier can be used to reason about lens equivalence. Similarly, we can define an order on scenes:

\begin{definition}[Scene Order] $F$ is a subscene of $G$, written $F \le G$, provided that $\mrg[F]~(\mrg[G]~s_1~s_2)~s_3 = \mrg[G]~s_1~(\mrg[F]~s_2~s_3)$ for any $s_1, s_2, s_3 \in \src$. \isalink{https://github.com/isabelle-utp/Optics/blob/f48265c6c05a4c7dc6951ecc29815369b8ae9377/Scenes.thy\#L331}
\end{definition}
The scene order is reflexive, transitive, and antisymmetric and thus forms a partial order on scenes. It also preserves state equivalence: if $s_1 \sceneq{G} s_2$, and $F \le G$, then also $s_1 \sceneq{F} s_2$. Moreover, we have $X \lsublens Y$ if and only if  $X^\sim \le Y^\sim$, which validates our order on scenes.

Scenes are intended to model not just individual lenses, but also sets of lenses. We therefore next define a set of Boolean algebra operators for scenes, including complementation.
\begin{definition}[Scene Boolean Lattice] \isalink{https://github.com/isabelle-utp/Optics/blob/f48265c6c05a4c7dc6951ecc29815369b8ae9377/Scenes.thy\#L144} $ $%
\vspace{1ex}

\centering
$\begin{array}{rlcrl}
\bot &~=~ (\lambda s_1~s_2.\, s_1)  &~~&   \top &~=~ (\lambda s_1~s_2.\, s_2) \\
F^\complement &~=~ (\lambda s_1~s_2.\, F~s_2~s_1) &~~& F \sqcup G &~=~ (\lambda s_1~s_2.\, G~(F~s_1~s_2)~s_2)
\end{array}$
\end{definition}
The $\bot$ scene is the merge function that ignores its second argument, and so takes none of $s_2$. Intuitively, $\bot$ is an empty set of lenses. Conversely, $\top$ ignores the first argument, and so takes all of the state from $s_2$. It therefore corresponds to the set of all consistent lenses. The complement operator, $F^\complement$, switches the two arguments of the merge function. It corresponds to the set complement, an operator that cannot easily be defined for lenses, due to the need to calculate the complement view~\cite{Hofmann2011SymmetricLenses}. The join operator, $F \sqcup G$, combines the two underlying merge functions. It first merges $s_1$ and $s_2$ using $F$, and then merges the resulting state with $s_2$ using $G$. We can also define a meet operator, $F \sqcap G \defs (F^\complement \sqcup G^\complement)^\complement$.

Intuitively, $\sqcup$ corresponds to the set union operator, and $\sqcap$ is intersection. These intuitions are partly justified and demonstrated in the following algebraic laws:

\begin{theorem}[Scene Laws] \isalink{https://github.com/isabelle-utp/Optics/blob/f48265c6c05a4c7dc6951ecc29815369b8ae9377/Scenes.thy\#L185}
$$
\top^\complement = \bot \qquad (F^\complement)^\complement = F \qquad \top \sqcup F = F \sqcup \top = \top  \qquad \bot \sqcup F = F \sqcup \bot = F \qquad F \sqcup F^\complement = \top
$$
$$
\bot \le F \qquad F \le \top \qquad F \lindep F^\complement \qquad F \lindep G^\complement \longleftrightarrow F \le G \qquad s_1 \sceneq{\bot} s_2
$$
\end{theorem}
We also show correspondences between the lens combinators and scenes:

\begin{theorem}[Lenses as Scenes] \isalink{https://github.com/isabelle-utp/Optics/blob/f48265c6c05a4c7dc6951ecc29815369b8ae9377/Scenes.thy\#L472}
$$
        \lone^\sim = \top \quad
        \lzero^\sim = \bot \quad
        (X \lplus Y)^\sim = X^\sim \sqcup Y^\sim \textit{ for } X \lindep Y \quad
        s_1 \sceneq{X^\sim} s_2 \textit{ iff } \lget_X~s_1 = \lget_X~s_2
$$
\end{theorem}

\noindent The scene of the identity lens is the top scene, since $\lone$ covers the whole state space. Correspondingly, the empty lens, $\lzero$, views none of the state space. The summation of lenses $X$ and $Y$ yields the union of the two underlying scenes, provided $X \lindep Y$. Finally, two states $s_1$ and $s_2$ are equivalent modulo lens $X$ precisely when applying the underlying $\lget_X$ function yields the same value for each state.

In order to substantiate the intuition of scenes as sets of variables, we would reasonably expect that scenes form a complete Boolean algebra. However, this is not the case for the set of all scenes, $Scene(\src)$. For example, we cannot prove that $\sqcup$ is associative, and several other properties do not hold. We therefore need to place additional restrictions on the space of scenes. 
This is because several valid state spaces over $\src$ may exist. We can characterise whether two scenes address consistent coordinate systems using scene compatibility.

\begin{definition}
    Scenes $F$ and $G$ are compatible, written $F \scompat G$, if $F \sqcup G = G \sqcup F$. \isalink{https://github.com/isabelle-utp/Optics/blob/f48265c6c05a4c7dc6951ecc29815369b8ae9377/Scenes.thy\#L122}
\end{definition}
Compatibility is specified as commutativity of scene union. This means that the merge functions can be applied in either order with the same result, that is, $G~(F~s_1~s_2)~s_2 = F~(G~s_1~s_2)~s_2$ for any given states $s_1$ and $s_2$. The intuition here is that $F$ and $G$ both use the same coordinate system, and one does not interfere with the other. This is the case if $F$ are $G$ are independent, and so we have the result $F \lindep Q \implies F \scompat G$. But it is not necessary for scenes to be independent to be compatible, for example, every scene is self-compatible, $F \scompat F$. Similarly, $F \scompat F^\complement$, $\top$ and $\bot$ are both compatible with any scene, and compatibility is closed under $\sqcup$. Moreover, $\sqcup$ is associative when applied to a set of pairwise compatible scenes. We also have the following important result:

\begin{theorem} \label{thm:scene_union_eq} Given scenes $F \scompat G$,
    $s_1 \sceneq{(F \sqcup G)} s_2 \longleftrightarrow (s_1 \sceneq{F} s_2 \land s_1 \sceneq{G} s_2)$. \isalink{https://github.com/isabelle-utp/Optics/blob/f48265c6c05a4c7dc6951ecc29815369b8ae9377/Scenes.thy\#L627}
\end{theorem}
This theorem allows us to decompose state equality according to the scene union operator, provides that the two scenes $F$ and $G$ are compatible. This theorem allows us to partition the state, and will form the basis for many of the results in Section~\ref{sec:programs}.

However, compatibility of scenes is not a given. For example, consider the Cartesian coordinate lenses $xc$ and $yc$ and the polar coordinate lenses $\textit{radius}$ and $\textit{angle}$ from Example~\ref{ex:polarlens}. These two pairs of lenses are fundamentally incompatible with each other since, they provide two different views on the underlying state space. Therefore, we can derive $\neg (\textit{xc} \scompat \textit{radius})$, for example, since these two lenses interfere with each other's view. Changing the $x$-coordinate affects the radius in the polar coordinate system, and vice-versa. Therefore, when speaking of a state space, we will be identifying a family of compatible scenes.

\section{Scene Spaces}
\label{sec:scene_spaces}

In this section, we introduce scene spaces, which allow us to characterise consistent subsets of $Scene(\src)$ that satisfy the laws from Boolean algebra and the algebra of sets. We do this by restriction to subsets of scenes that are pairwise compatible, and thus address a consistent coordinate system. To introduce scene spaces, we first need to introduction of the scene basis, which allow us to identify the fundamental set of atomic variables:

\begin{definition}[Scene space basis]
A sequence of scenes $S_i \in Scene(\src),\, i \in \{1 \cdots n\}$ for some nonempty set $\src$ forms a scene space basis $\basis$ if it satisfies the following conditions: \isalink{https://github.com/isabelle-utp/Optics/blob/f48265c6c05a4c7dc6951ecc29815369b8ae9377/Scene_Spaces.thy\#L216}
$$\forall\, i,j \le n,\ i \neq j \implies S_i \lindep S_j \textnormal{  (independence); and } S_1 \sqcup S_2 \sqcup \cdots \sqcup S_n = \top \textnormal{ (coverage)}$$
\end{definition}

\noindent We define $\basis = \{S_i ~|~ i \in \{1 \cdots n\}\}$ as the set of basis scenes. In the context of program verification, we also call the basis scenes the variables, or \textit{Vars}, of the program. The use of a sequence here hints at the needs of the underlying mechanisation --- a sequence lets us ensure that the basis is finite, that it can be used for computation (due to the natural mapping to lists), and gives us a unique, canonical label for each scene in the basis.

Intuitively, the set of basis scenes is the collection of atomic state components, or the set of named variables --- each such component uniquely identifies a coordinate in the state space. The conditions imposed on basis scenes ensures that they do not overlap (independence), and that together they cover the entire space (coverage) so that every part of the state can be updated with one of the components. Since every lens can be assigned a scene, we can also identify the set of \textit{basis lenses}, as those lenses that have a corresponding basis scene.

A scene space is the smallest set $\mathfrak{S}_\mathcal{B}$, which contains the basis and is closed under union.

\begin{definition}[Scene space]
A scene space over $\basis$ is a finite set $\mathfrak{S}_\mathcal{B}$ such that: \isalink{https://github.com/isabelle-utp/Optics/blob/f48265c6c05a4c7dc6951ecc29815369b8ae9377/Scene_Spaces.thy\#L229}
    $$\bot \in \mathfrak{S}_\mathcal{B}; \quad 
    \text{If } A \in \mathcal{B} \text{ then } A \in \mathfrak{S}_\mathcal{B}; \quad
    \text{If } A, B \in \mathfrak{S}_\mathcal{B} \text{ then } A \sqcup B \in \mathfrak{S}_\mathcal{B}.$$
\end{definition}

\noindent We encode this in Isabelle as an inductive predicate, which provides an inductive reasoning principle for reasoning over elements of $\mathfrak{S}_\mathcal{B}$. If we consider the basis scenes as individual, atomic variables, then each $A \in \mathfrak{S}_\mathcal{B}$ corresponds to some set of variables, including the empty set ($\bot$). All pairs of scenes within a scene space satisfy the core property of compatibility:

\begin{lemma}
For any $A, B \in \mathfrak{S}_\mathcal{B}$, we have $A \scompat B$. \isalink{https://github.com/isabelle-utp/Optics/blob/f48265c6c05a4c7dc6951ecc29815369b8ae9377/Scene_Spaces.thy\#L276}
\end{lemma}

\begin{proof}
 By induction on the scene space structure of $A$ and $B$.
\end{proof}

\noindent The proof relies on the fact that scene union preserves compatibility, and the basis scenes are all compatible due to their pairwise independence. Given this property, we can see that scene spaces are better behaved than the set of all scenes of a given type. In particular, we now have commutative unions, which allows defininition of the following operator:
\begin{definition}[Scene Set Union]
Given some finite set of scenes $\textbf{S} = \{S_1, S_2, \dots, S_n\}$ in some scene space $\mathfrak{S}_\mathcal{B}$ we define $\bigcup \textbf{S} = S_1 \sqcup S_2 \sqcup \dots \sqcup S_n$ \isalink{https://github.com/isabelle-utp/Optics/blob/f48265c6c05a4c7dc6951ecc29815369b8ae9377/Scene_Spaces.thy\#L16}
\end{definition}

\noindent This works analogously to the standard finite set union operator. The chosen ordering given to the scenes in the finite set $\textbf{S}$ does not matter as scene union commutes within a scene space. With this, we can define the decomposition of scenes into their component basis scenes. This important result will allow us to justify our intuition of scenes-as-sets:
\begin{lemma}[Basis decomposition]
Given some scene $A$ in scene space $\mathfrak{S}_\mathcal{B}$, its decomposition $\mathcal{D}(A) \in \mathcal{P}(Scene(\mathcal{S}))$ is the unique set of scenes satisfying the following properties: \isalink{https://github.com/isabelle-utp/Optics/blob/f48265c6c05a4c7dc6951ecc29815369b8ae9377/Scene_Spaces.thy\#L771}
    $$\mathcal{D}(A) \subseteq \basis \qquad \bigcup \mathcal{D}(A) = A \qquad \bot \notin \mathcal{D}(A)$$
\end{lemma}
Thus, we can decompose any $A \in \mathfrak{S}_\mathcal{B}$ into a set of basis scenes. The decomposition is unique as every element of $\mathfrak{S}_\mathcal{B}$ is built by composition of basis scenes, and we exclude the bottom scene. $\mathcal{D}$ is an injective function on $\mathfrak{S}_\mathcal{B}$, so we can demonstrate equality of two scenes by showing these decompose to the same set. The decomposition gives us a way to relate set operations to scene operations, and we find a direct relation between them: \isalink{https://github.com/isabelle-utp/Optics/blob/f48265c6c05a4c7dc6951ecc29815369b8ae9377/Scene_Spaces.thy\#L860}
$$
    \mathcal{D}(\top) = \mathcal{B} \setminus \{\bot\} \qquad
    \mathcal{D}(A \sqcup B) = \mathcal{D}(A) \cup \mathcal{D}(B) \qquad
    \mathcal{D}(A \sqcap B) = \mathcal{D}(A) \cap \mathcal{D}(B) \qquad
$$
$$
    \mathcal{D}(A^\complement) = \mathcal{B} \setminus \mathcal{D}(A) \setminus \{\bot\} \qquad
    (A \subseteq B) \iff \mathcal{D}(A) \subseteq \mathcal{D}(B)
$$
The equalities demonstrate that $\mathcal{D}$ is a homomorphism between scene spaces and sets, and therefore justifies our use of scenes to characterise sets of variables. We can define the distributed scene intersection by taking advantage of this analogy: $\bigcap_{i=1}^n S_i \defs (\bigcup_{i=1}^n S_i^\complement)^\complement$.

%
This lattice structure allows us to prove the following properties of the $\bigcup$ operator:
\begin{lemma}
    $
        s_1 \approx s_2 \text{ on } \bigcup X \iff  (\forall\, x \in X.\, (s_1 \approx s_2 \text{ on } x))
    $ \isalink{https://github.com/isabelle-utp/Optics/blob/f48265c6c05a4c7dc6951ecc29815369b8ae9377/Scene_Spaces.thy\#L940}
\end{lemma}
This is a generalisation of Theorem~\ref{thm:scene_union_eq}, but for a subset $X \subseteq \mathfrak{S}_\mathcal{B}$. Similarly, we can prove the following law relating the set complement $\mathcal{B}$ and the scene complement:

\begin{lemma}[Scene Complement Union]
    Given $X \subseteq \basis$, the union of the set of all the variables $\basis$ minus $X$ is the complement of the union of $X$:
    $\bigcup(\basis\setminus X) = (\bigcup X)^\complement$
\end{lemma}
With this, we can show that scene spaces form a complete lattice, where the scene set union acts as a supremum, and also a complete Boolean algebra, with $A^\complement$ as complementation:

\begin{theorem}[Scene Set Complete Boolean Algebra]
The operator $\bigcup$ together with the ordering $\subseteq$ and complement $A^\complement$ induces a complete Boolean algebra on the scene space. \isalink{https://github.com/isabelle-utp/Optics/blob/f48265c6c05a4c7dc6951ecc29815369b8ae9377/Scene_Spaces.thy\#L1321}
\end{theorem}

Thus, with scene spaces we have arrived at sufficient structure on a source $\src$ to characterise the variables in the state space, the relationship between them (independence and overlap), the characterisation of a state space (cf. Definition~\ref{def:state-space}) and the ability to collect them in sets. In the next section we will put this theory to work in describing the meta-logical operators of a programming reasoning calculus.


\section{Reasoning about Programs}
\label{sec:programs}
Through scene spaces, we are equipped with a sufficient algebraic structure to develop meta-logical functions on expressions and programs. Our development is based on the Isabelle/UTP program model~\cite{Foster2020-IsabelleUTP}, which follows the programs-as-predicates approach~\cite{Hehner93,Hoare&98}. Isabelle/UTP is an implementation of Hoare and He's \textit{Unifying Theories of Programming}~\cite{Hoare&98} semantic framework in Isabelle/HOL. 

As usual in a shallow embedding, expressions are modelled as functions $e : \src \to V$. Application of a function, $e(s)$, corresponds to evaluation of the expression in a particular state $s \in \src$. For example, given lenses $x, y, z : \mathbb{Z} \lto \src$, we denote the expression $x + y \le z$ (syntax) by the expression function $\lambda s.\, \lget_x~s + \lget_y~s \le \lget_y~s$ (semantics). Our Shallow Expressions library manages the automatic translation between these two representations~\cite{Shallow_Expressions-AFP}.

Although expression functions represent only semantics, without explicit abstract syntax tree, scenes allows us to characterise meta-logical properties. The unrestriction operator $a \unrest e$~\cite{Oliveira07} describes dependence of an expression on scene $a$. If $e$ does not make use of any of the variables of $a$ in its evaluation, then $a$ is unrestricted, that is $a \unrest e$. This operator has been formally defined over lenses, but we can now generalise it for scenes.
\begin{definition}[Unrestriction] $(a \unrest e) \defs (\forall s~s'.\, e~(\mrg[a]~s~s') = e(s))$
\end{definition}
The expression $e$ is unrestricted by the variables of $a$ if whenever we overwrite the values with those from another state $s'$, the evaluation of $e$ does not change. With this definition, we can recover the more specific definition of unrestriction for lenses~\cite{Foster2020-IsabelleUTP}:

\begin{theorem} $x^\sim \unrest e \leftrightarrow (\forall s~v.\, e~(\lput_x~s~v) = e~s)$
\end{theorem}
An expression is unrestricted by a lens, $x^\sim \unrest e$, if evaluating the expression in a state where $x$ has been modified does not alter the evaluation.
Proving an unrestriction conjecture for a lens is, on the most part, fully automated. In Isabelle this process reduces to $\beta$-reduction and substitution. We can show that if substitution of a lens with an arbitrary value does not change the valuation of the expression, then the variable is unrestricted, e.g. $(x + y)[z \mapsto v] = x + y$, and so $z \unrest (x + y)$. However, checking that an expression is \emph{not} unrestricted by $x$ is, in general, undecidable. We cannot simply check for the presence of a variable, since this does not guarantee a dependency. For example, the expression $x \le 0 \lor x > 0$ reduces to $\textit{True}$ and so does not depend on $x$, even though it is mentioned.

Unrestriction has a deep link with the free variables of an expression. We take the following definition from Andr\'{e} Platzer~\cite[Definition 8]{platzer_complete_2017}, used in his foundational work on a uniform substitution principle for differential dynamic.
Here, the free variables are those that influence the evaluation of an expression. With scenes, we define:
\begin{definition}
    $
        FV(e) \defs  \bigcup\,\{x\in \basis\, \vert\, \exists\, s,s': s\approx s' \text{ on } x^\complement\, \land \,  e(s) \neq e(s')\}
    $\isalink{https://github.com/cisasam/IsabelleFreeBoundVariables/blob/main/Vars_UTP.thy\#L26}
\end{definition}
The set of free variables of $e$ is the subset of the basis scenes where there are two states $s$ and $s'$ that are equal except on $x$, but disagree on their evaluation of $e$, which implies a dependency. 
For example, $FV(x + y)=\{x,y\}$. From this, we can prove the following lemma:

\begin{lemma}[Coincidence Effect for Terms and Formulas]
    The set of free variables of a term/formula $FV(e)$ is the smallest set that holds the coincidence effect property: \isalink{https://github.com/cisasam/IsabelleFreeBoundVariables/blob/main/Vars_UTP.thy\#L138}
    \begin{align*}
        s \approx s' \text{ on } \bigcup FV(e) \implies  e(s) = e(s')
    \end{align*}
\end{lemma}
Which ensures that two states that agree on the free variables of an expression will agree on the evaluation of that expression.
Its usefulness comes from the fact that if all the end states of two programs always agree on the values of the free variables of an expression, then the evaluation of that expression will be the same after any of the two programs.

We can restate $FV$ using unrestriction:
\begin{lemma}[Not Unrestricted is Free]
    $
        FV(e) = \{x\in\mathcal{V}\,\vert\, \neg (x\, \sharp\, e)\}
    $\isalink{https://github.com/cisasam/IsabelleFreeBoundVariables/blob/main/Vars_UTP.thy\#L6}
\end{lemma}
Which implies that checking $x \notin FV(e)$ is decidable, since unrestriction is decidable. 

We can use scene spaces to define programs as predicates $P(s,s')$ with $s$ and $s'$ being states.
This predicate is true if $(s,s')$ is a valid transition of the program $P$, i.e. there is a valid run of program $P$ with initial state $s$ and final state $s'$. Ensuring that a program $P$ holds a property $F$, specified as a boolean expression, is at the core of any verification problem.
We will use the modal box operator to define this concept:
\begin{definition}[Box Operator]
    We say that property $F$ holds for program $P$, $[P]F$, if every final state of a valid run of $P$ holds $F$, that is $[P]F \iff \forall s'\;P(s,s'): F~s'$
\end{definition}

Free and bound variables can be used to reason about programs and properties. We lift the definitions of free and bound variables of a program from~\cite[Definition 8]{platzer_complete_2017}.
These definitions can be mechanised using scenes:
\begin{definition}\label{def:fv-progs}
    $
        FV(P) \defs \bigcup\,\{x\in \basis\, \vert\, \exists\, s_0,s_0',s_1: s_0\approx s_1 \text{ on } x^\complement \land P(s_0, s_0')\,\land 
        (\nexists\, s_1': s_0'\approx s_1'\text{ on } x^\complement \land P(s_1,s_1')) \}
    $\isalink{https://github.com/cisasam/IsabelleFreeBoundVariables/blob/main/Vars_UTP.thy\#L33}
\end{definition}
The free variables of a program are the ones that influence the runs of that program, i.e. the variables whose values influence the values of other variables after the program ends.
\begin{definition}\label{def:bv-progs}
    $
        BV(P) \defs  \bigcup\,\{x\in \basis\, \vert\, \exists\, s_0,s_0': P(s_0, s_0')\land \neg(s_0 \approx s_0' \text{ on } x) \}
    $\isalink{https://github.com/cisasam/IsabelleFreeBoundVariables/blob/main/Vars_UTP.thy\#L39}
\end{definition}
On the other hand, the bound variables of a program are the ones whose value changes through a run of a program. Using these definitions we can demonstrate a simple example:
\begin{example}
    Given $P\triangleq \{x:= x+1\}$, $BV(P)=\{x\}$ and $FV(P)=\{\}$.
    \isalink{https://github.com/cisasam/IsabelleFreeBoundVariables/blob/main/Vars_UTP_Example.thy\#L26}

\end{example}
As $x$ changes value through the execution of the program and there are no variables that influence other variables ($x$ only influences itself).
Furthermore, we can prove that free variables are the ones that influence the output of a program.
\begin{lemma}[Coincidence Effect Lemma for Programs]
    A set of variables $S$ satisfies the coincidence effect property of $P$ if $s_0 \approx s_1$ on $\bigcup V \supseteq S$ and $P(s_0,s_0')$, then there is an $s_1'$ such that $P(s_1,s_1')$ and $s_0'\approx s_1'$ on $V$.
    $FV(P)$ is the smallest such set: $FV(P) \subseteq S$. \isalink{https://github.com/cisasam/IsabelleFreeBoundVariables/blob/main/Vars_UTP.thy\#L257}
\end{lemma}
This lemma states that if two initial states agree on the free variables of a program, for any final state of one we can find a final state of the other that agrees on any superset of the free variables, including the entire variable space.
The proof of the lemma is split into two parts: proving that $FV(P)$ holds the coincidence lemma and proving that it is the smallest set.
For the first part, we can prove it using induction over some intermediate states (following the original proof~\cite{platzer_complete_2017}).
For the second, it can be proven using \cref{def:fv-progs}. As an example:

\begin{example}
    $\{\}$ satisfies the coincidence effect property for $P\triangleq \{x:=x+1\}$. \isalink{https://github.com/cisasam/IsabelleFreeBoundVariables/blob/main/Vars_UTP_Example.thy\#L54}
\end{example}
As the $FV(P)=\{\}$, it holds per the coincidence effect lemma.

We can also prove the fact that the initial and final state of any valid transition of a program will agree at least in all non-bounded variables. 
\begin{lemma}[Bound Effect Lemma for Programs]\label{beffect}
    The set of bound variables of a program $BV(P)$ is the smallest set that satisfies the bound effect property: \isalink{https://github.com/cisasam/IsabelleFreeBoundVariables/blob/main/Vars_UTP.thy\#L59}
    \begin{align*}
        P(s,s') \implies s \approx s' \text{ on } (\bigcup\nolimits BV(P))^\complement
    \end{align*}
\end{lemma}
Which can be proven using \cref{def:bv-progs}.
As an example:
\begin{example}
    The set $\{x\}$ satisfies the bound effect property for $P\triangleq \{x:=x+1\}$.
    \isalink{https://github.com/cisasam/IsabelleFreeBoundVariables/blob/main/Vars_UTP_Example.thy\#L108}
\end{example}
Which holds as $\{x\}$ corresponds to the bound variables of $P$.

Using these two lemmas and the definitions of bound and free variables we can prove foundational properties that are widely applied in program verification.
For example:
\begin{lemma}
    Given two programs $P$ and $Q$, if $BV(P)\cap BV(Q) = \emptyset$, $FV(P)\cap BV(Q) = \emptyset$ and $BV(P)\cap FV(Q) = \emptyset$, then $P;Q = Q;P$.
    \isalink{https://github.com/cisasam/IsabelleFreeBoundVariables/blob/main/Vars_UTP.thy\#L286}
\end{lemma}
As the free variables are the ones that influence the runs of the programs and the bound variables are the ones that change, if these two sets are pairwise disjoint, the changes that $P$ makes do not influence $Q$ and vice-versa.
Furthermore, since the bound variables are also disjoint, $Q$ does not overwrite the changes that $P$ makes and vice-versa.

In order to further showcase the power of scenes we define a simple composition operator and prove some properties about it.
This operator, $P\parallel Q$, merges over the bound variables of the different processes in the parallel composition.
We are going to use a simplified version of the one defined in~\cite{brieger2025completedynamiclogiccommunicating} with no trace or communication.
We will prove that it holds two desirable properties: parallel injection and parallel decomposition

First, parallel injection states that if $[P]F$ and $Q$ does not-interfere (i.e. runs of this program do not affect the truth value of this program), then $[P\parallel Q] F$.
We can can characterize non-interference using free and bound variables:
\begin{definition}[Non-interference]
    Program $P$ does not interfere with formula $F$ if $BV(P)\cap FV(F)=\emptyset$
    \isalink{https://github.com/cisasam/IsabelleFreeBoundVariables/blob/main/Vars_UTP.thy\#L538}
\end{definition}
The set $BV(P)$ contains the variables that can be affected by running $P$ while $FV(F)$ contains the variables that influence the value of $F$, so if these two sets are disjointed, we achieve what we want, runs of $P$ will not influence the value of $F$.

Second, parallel decomposition states the fact that if $P$ does not interfere with $F_Q$ and $Q$ does not interfere with $F_P$, then $[P]F_P$ and $[Q]F_Q$ are sufficient to state $[P\parallel Q](F_P\land F_Q)$.
Parallel decomposition is a powerful property to have, as it allows us to decompose safety properties of the parallel composition into properties that hold for the partial programs and decompose a proof of safety for the composition into one of the partial programs.

We define the parallel composition as follows:
\begin{definition}[Simple Parallel Operator]\label{def:parallel}
    The parallel operator, $P\parallel Q$ is a parallel-by-merge operator~\cite{Hoare&98} that merges $P$ and $Q$ over the bound variables of $P$ \isalink{https://github.com/cisasam/IsabelleFreeBoundVariables/blob/main/Vars_UTP.thy\#L401}
    \begin{align*}
        (P\parallel Q) (s,s') \iff  \exists~s_P, s_Q.~s'=\mrgon{s_Q}{s_P}{\bigcup BV(P)}\land P(s,s_P)\land Q(s,s_Q)
    \end{align*}
\end{definition}
Here, $\mrgon{s_1}{s_2}{A}$ is notation for $\mrg[A]~P~Q$.
This parallel operator splits the state into two, the variables bounded by $P$ and the rest. 
It runs both programs independently and then it merges the state, taking the values of the bound variables of $P$ from $P$ itself and the rest from $Q$.

For example, given variables $x,y,z$ and the program $P\triangleq\{x:=z\parallel y:=z\}$ that assigns $x$ and $y$ to $z$ in parallel, it will have as valid runs any pair of states $(s,s')$ where $s'[x]=s[z]$ and $s'[y]=s[z]$.

Our parallel operator commutes given that the programs do not share bound variables:
\begin{lemma}[Commutativity of $\parallel$]
    Given that $BV(P)\cap BV(Q) =\emptyset$,  $P \parallel Q = Q\parallel P$ 
    \isalink{https://github.com/cisasam/IsabelleFreeBoundVariables/blob/main/Vars_UTP.thy\#L632}
\end{lemma}
This holds by the bound effect lemma (\cref{beffect}) and the definition of the parallel operator.
As $P$ and $Q$ are merged over the bound variables of $P$, and as $Q$ can only affect variables in $BV(Q)$, if these two sets are disjoint the operator commutes.

In order to prove both parallel injection and decomposition we need to discuss the final states of the programs that make up the parallel composition.
We can mechanise them using scenes:
\begin{definition}[Partial Left and Partial Right]
    For a valid run of the parallel composition $P\parallel Q\,(s,s')$, we define a partial left state, $s_P$ as any state that holds
$P(s,s_P)\land s' \approx s \text{ on } \bigcup\nolimits BV(P)$ and a partial right state $s_Q$ as any state that holds $Q(s,s_Q)\land s' \approx s \text{ on } \left(\bigcup\nolimits BV(P)\right)^\complement$
\isalink{https://github.com/cisasam/IsabelleFreeBoundVariables/blob/main/Vars_UTP.thy\#L405}
\end{definition}

With these definitions we can reason about the final states of the parallel composition using the final states of the programs that take part in the composition.
Specifically, we can show the following result:

\begin{lemma}
    For all $s$, $s'$ and $s_P$ s.t. $(P \parallel Q) (s, s')$ and $s_P$ is a partial left, we have: \isalink{https://github.com/cisasam/IsabelleFreeBoundVariables/blob/main/Vars_UTP.thy\#L524}
    \begin{align*}
        s'\approx s_P \text{ on } \left(\bigcup BV(Q)\right)^\complement
    \end{align*}
\end{lemma}
First by the definition of the parallel operator, $s'\approx s_P$ on $BV(P)$ so the states also agree on $BV(P)\cap (BV(Q))^\complement$.
Second, $s \approx s_P$ on $(BV(P))^\complement$ by the bound effect lemma, and for all partial right states $s_Q$ we can make the same argument $s \approx s_Q$ on $(BV(Q))^\complement$, which means that $s_P\approx s_Q$ on $(BV(P))^\complement \cap (BV(Q))^\complement$.
Finally, $s_Q\approx s'$ on $(BV(P))^\complement$ by the definition of the parallel operator, so $s_P\approx s'$ on $(BV(P))^\complement \cap (BV(Q))^\complement$.
Putting everything together we have $s_P\approx s'$ on $(BV(P)\cap (BV(Q))^\complement)\cup ((BV(P))^\complement \cap (BV(Q))^\complement) = (BV(Q))^\complement$.
And as a direct consequence we have
\begin{lemma}[Non-interference Preserves Safety]
    For any pair of states belonging to the parallel composition $(P \parallel Q) (s, s')$ with partial left $s_P$, given that $Q$ does not interfere with formula $F$, $s'\llbracket F\rrbracket$ holds if and only if $s_P\llbracket F\rrbracket$.
    \isalink{https://github.com/cisasam/IsabelleFreeBoundVariables/blob/main/Vars_UTP.thy\#L573}
\end{lemma}
As due to non-interference $\bigcup\nolimits FV(F)\leq \left(\bigcup\nolimits BV(Q)\right)^\complement$ and the lemma holds by coincidence.

Finally, using the box operator and the previous lemma we can obtain the parallel injection property:
\begin{theorem}[Parallel Injection]
    Given that $Q$ does not interfere with formula $F$, $[P]F\implies [P\parallel Q]F$
    \isalink{https://github.com/cisasam/IsabelleFreeBoundVariables/blob/main/Vars_UTP.thy\#L602}
\end{theorem}
The power of this theorem is tied to the fact that it allows us to verify properties about programs in a parallel composition without having to use the definition of the operator:
\begin{example}
    Let $P\triangleq \{x:=y+1\}$, $Q\triangleq \{z:=y+1\}$ and $F\triangleq x>y$. $Q$ does not interfere with $F$ as $BV(Q)=\{z\}$ and $FV(F)=\{x,y\}$ and $[P]F$ trivially, then $[P\parallel Q]F$.
\end{example}

Finally, with the previous lemma and commutativity:
\begin{theorem}[Parallel Decomposition]
    Given that $P$ does not interfere with formula $F_Q$ and $Q$ does not interfere with formula $F_P$, $[P]F_P\land[Q]F_Q\implies [P\parallel Q](F_P\land F_Q)$
    \isalink{https://github.com/cisasam/IsabelleFreeBoundVariables/blob/main/Vars_UTP.thy\#L716}
\end{theorem}
With this last theorem, we can conclude that scenes, via the definitions of free and bound variables, allow us to reason about parallel programs (defined through \cref{def:parallel}) compositionally.
Further examples on how to use free and bound variables to enable compositional reasoning can be seen in the proofs of soundness of $\textsf{dL}_\textnormal{CHP}$~\cite{brieger2025completedynamiclogiccommunicating} and \textsf{d}$\mathcal{L}$~\cite{platzer_complete_2017}.


\section{Related Work}
\label{sec:related}
We previously surveyed different approaches to state space modelling in proof assistants~\cite[Section~2.4]{Foster2020-IsabelleUTP}, using Schirmer and Wenzel's categories as a framework~\cite{Schirmer2009}. They identify four approaches employed in developing verification tools: state as a (1) function; (2) tuple; (3) record; and (4) abstract type. Lenses unify these different approaches by providing an algebraic account of state that does not depend on a particular model~\cite{Foster2020-IsabelleUTP}. In particular, Isabelle/UTP is a confluence of ideas from Oliveira~\cite{Oliveira07} et al. and Feliachi et al.~\cite{Feliachi2010} in their previous respective mechanisations of the UTP. Scenes build on this by allowing lenses to be collected in sets, and without any loss of generality in our previous work. Crucially, scenes avoid the need to formalise (1) names for variables and (2) a value universe for the state model, which is necessary for a shallow embedding.

The term ``lens'' originates in the work of Nate Foster~\cite{Foster07}, and lenses are now ubiquitous in functional programming. The idea of using algebra to provide a foundation for state space modelling is well-known in the literature, going back to the work of Tarski~\cite{Tarski71} and Oles~\cite{Oles82}. Notably, Back and von Wright~\cite{Back1998} use essentially the same algebraic structure~\cite{Foster2020-IsabelleUTP} to lenses to model variables in their refinement calculus. Hayes~\cite{Hayes2016a,Colvin2017} et al. have also employed similar algebraic properties, including unrestriction, in the context of rely-guarantee reasoning for shared variable concurrency. Also in this line, Dongol et al.~\cite{Dongol19} use a functional state space and mechanised algebra to implement cylindrical algebra in Isabelle/HOL, which is then extended for programs and used to model concepts like local variables. Similarly, Platzer uses the algebraic properties of variables to formulate properties of hybrid programs for his uniform substitution calculus~\cite{platzer_complete_2017}.  Our work on lenses and scenes allows these ideas to be generalised and put to work in the context of verification tools in proof assistants. 

\section{Conclusions}
\label{sec:concl}
We have introduced scenes as a meta-logical algebra for reasoning about mutable state. Though based on an elementary algebraic foundation, they allow us to model sets of variables in the context of a shallow embedding, and thus allow us to support laws of programming that depend on the program frame. These results on scenes form part of our \textsf{Optics} library~\cite{Optics-AFP} for Isabelle/HOL, and our Isabelle/UTP relational program model~\cite{Foster2020-IsabelleUTP}.

In future work, we are using our mechanisation of $BV$ and $FV$ as a foundation to implement Brieger's $\textsf{dL}_{\textnormal{CHP}}$~\cite{BriegerMP23} in Isabelle/HOL. In parallel with this, we are currently developing automation support for generation of scene spaces in Isabelle/HOL. For example, while calculation of $x \notin FV(e)$ can be automated through unrestriction, calculating $FV(e)$ itself requires further proof methods. In a different axis, we are also currently developing a port of the Optics library in Lean, so that our results can have more reach.

\bibliography{ref}
\end{document}